\documentclass[english, 11pt]{article}
\usepackage{amsfonts}
\usepackage{amssymb}
\usepackage{graphicx}
\usepackage{amsmath}

\newtheorem{theorem}{Theorem} [section]

\newtheorem{corollary}[theorem]{Corollary}

\newtheorem{definition}[theorem]{Definition}
\newtheorem{example}[theorem]{Example}

\newtheorem{lemma}[theorem]{Lemma}

\newtheorem{proposition}[theorem]{Proposition}

\newenvironment{proof}[1][Proof]{\textbf{#1.} }{\ \rule{0.5em}{0.5em}}
\allowdisplaybreaks[1]

\usepackage{color}
\usepackage[normalem]{ulem}
\newcommand{\JP}[2]{{\color{blue}#1\ifmmode\msout{#2}\else\sout{#2}\fi}}
\newcommand{\AM}[2]{{\color{green}#1\ifmmode\msout{#2}\else\sout{#2}\fi}}
\newcommand{\LG}[2]{{\color{red}#1\ifmmode\msout{#2}\else\sout{#2}\fi}}

\begin{document}

\author{Luis A. Guardiola\thanks{%
Luis A. Guardiola \newline
Departamento de Fundamentos del An\'{a}lisis Econ\'{o}mico, Universidad de
Alicante, Alicante 03071, Spain. E-mail: luis.guardiola@ua.es} \thanks{%
Corresponding author.}, Ana Meca\thanks{%
Ana Meca \newline
I.U. Centro de Investigaci\'{o}n Operativa. Universidad Miguel Hern\'{a}%
ndez, Edificio Torretamarit. Avda. de la Universidad s.n. 03202 Elche
(Alicante), Spain. E-mail: ana.meca@umh.es} and Justo Puerto\thanks{%
Justo Puerto \newline
Facultad de Matem\'{a}ticas, Universidad de Sevilla, 41012 Sevilla, SPAIN.
e-mail: puerto@us.es} }
\title{Cooperation in lot-sizing problems with heterogeneous costs: the
effect of consolidated periods}
\maketitle

\begin{abstract}
We consider a cooperative game defined by an economic lot-sizing problem
with heterogeneous costs over a finite time horizon, in which each firm
faces demand for a single product in each period and coalitions can pool
orders. The model of cooperation works as follows: ordering channels and
holding and backlogging technologies are shared among the members of the
coalitions. This implies that each firm uses the best ordering channel and
holding technology provided by the participants in the consortium. That is,
they purchase, hold inventory, pay backlogged demand and make orders at the
minimum cost of the coalition members. Thus, firms aim at satisfying their
demand over the planing horizon with minimal operation cost. Our
contribution is to show that there exist fair allocations of the overall
operation cost among the firms so that no group of agents profit from
leaving the consortium. Then we propose a parametric family of cost
allocations and provide sufficient conditions for this to be a stable family
against coalitional defections of firms. Finally, we focus on those periods
of the time horizon that are consolidated and we analyze their effect on the
stability of cost allocations. \medskip


\end{abstract}

\newpage

\section{Introduction}

Lot-sizing is one of the most important and also one of the most difficult
problems in production planning. Although lot sizing and scheduling
problems, and their variants, have been studied by many authors, providing
different solution approaches, looking for more efficient solution
approaches is still a challenging subject. Variants of the lot-sizing
problem (henceforth, LSP) with complex setup and other variants which are
more realistic and practical have received less attention in the literature.
There has been little literature regarding problems such as LSP with
backlogging or with setup times and setup carry-over. Since these problems
are NP-hard, fast and efficient heuristics are required. A good survey on
the subject is Karimi et al (2003).

Also there is little literature for cost-sharing in LSP. Among the pioneers
we mention Van Den Heuvel et al. (2007) which focuses on the cooperation in
economic lot-sizing situations (henceforth ELS-situations) with homogeneous
costs. They consider a homogeneous finite horizon model where cost are all
equal for all the players in each period. Each player must satisfy its
demand in each period by producing or carrying inventory from previous
stages but backlogging is not permitted. The main result in that paper is
that the cooperative games induced by ELS-situations enjoy a nonempty core.
Subsequently, Guardiola et al. (2008, 2009) present a new class of totally
balanced combinatorial optimization games: production-inventory games
(henceforth, PI-games). PI-games bring the essentials of inventory and
combinatorial optimization games. They provide a cooperative approach to
analyze the production and storage of indivisible items being their
characteristic function given as the optimal objective function of a
combinatorial optimization problem. PI-games can be seen as ELS games
without setup costs but with backlogging. Guardiola et al. (2008, 2009)
prove that the Owen set, the set of allocations which are achievable through
dual solutions [see Owen, 1975 and Gellekom et al., 2000] reduces to a
singleton. In addition, that allocation is always in the core and it defines
a population monotonic allocation scheme. This fact motivates the name of
Owen point for this core-allocation on PI-games. The main difference between
the ELS-games by Van Den Heuvel et al. (2007) and PI-games by Guardiola et
al. (2009) is that the former considers set up costs but assume that costs
are the homogeneous for all players in each period. Hence, both situations
(ELS and PI) are, in general, different.

On the other hand, Xu and Yang (2009) present a cost-sharing method that is
competitive, cross-monotonic and approximate cost recovering for an ELS-game
under a weak triangle inequality assumption, along with numerical results
showing the effectiveness of the proposed method. Li et al. (2014) present a
cost-sharing method that is cross-monotonic, competitive, and approximate
cost recovery, for the multi-level ELS-game, under a mild condition. This
result extends that of the recent 1-level ELS-game of Xu and Yang (2009).

Dreschel (2010) focusses on cooperative lot-sizing games in supply chains.
Several models of cooperation in lot-sizing problems of different complexity
that are analyzed regarding theoretical properties like monotonicity or
concavity and solved with the proposed row generation algorithm to compute
core elements; i.e. determining stable and fair cost allocations.

In another paper, Gopaladesikan and Uhan (2011) consider a new class of
cooperative cost-sharing games, associated with the ELS-problem with
remanufacturing options (henceforth, ELSR). They investigate the relative
strength and integrality gaps of various mathematical optimization
formulations of ELSR. Using insights from these results, they study the core
of the associated cost-sharing game and show it is empty, in general.
However, for two special cases- zero setup costs, and large initial quantity
of low cost returns- they find that the cost sharing game has a non-empty
core, and that a cost allocation in the core can be computed in polynomial
time.

Zeng et al. (2011) consider the ELS-game with perishable inventory. In this
cooperative game, a number of retailers that have a known demand through a
fixed number of periods for a same kind of perishable goods collaborate to
place joint orders to a single supplier. They first show that an ELS game
with perishable inventory is subadditive, totally balanced and its core is
non-empty. Then, they propose a core-allocation which allocates the unit
cost to each period as equally as possible. Finally, a numerical example is
given to illustrate the above results.

Tsao et al. (2013) use the Nash game and the cooperation game in an
imperfect production system to investigate the combined effects of
lot-sizing integration, learning effect, and an imperfect production process
on a manufacturer-retailer channel. They also developpe a search procedure
to solve the problem described, the optimal properties and a numerical study
are conducted to seek for structural and quantitative insights into the
relationship between the upstream and downstream entities of the supply
chain. Numerical results indicated that the cooperation game policy created
a higher cost reduction under a wide range of parameter settings.

Finally, Chen and Zhang (2015) consider the ELS-game with general concave
ordering cost. In that paper, the dual variables are understood as the price
of the demand per unit. They show that a core allocation can be computed in
polynomial time under the assumption that all retailers have the same cost
parameters (again homogeneous costs). Their approach is based on linear
programming (LP) duality. Specifically, they prove that there exists an
optimal dual solution that defines an allocation in the core and point out
that it is not necessarily true that every optimal dual solution defines a
core allocation. Toriello and Uhan (2014) also study ELS-games with general
concave ordering costs and show how to compute a dynamic cost allocation in
the strong sequential core of these games, i.e. an allocation over time that
exactly distributes costs and is stable against coalitional defections at
every period of the time horizon.

In this paper we study another class of totally balanced combinatorial
optimization games called setup-inventory games (henceforth, SI-games) that
arises from cooperation in lot-sizing problems with heterogeneous costs.
Each firm faces demand for a single product in each period and coalitions
can pool orders. Firms cooperate which implies that each firm uses the best
ordering channel and holding technology provided by the participants in the
consortium. That is, they purchase, hold inventory, pay backlogged demand
and make orders at the minimum cost among the members of the coalition.
Thus, firms aim at satisfying their demand over the planing horizon with
minimal operation cost.

Therefore, each firm must solve the Wagner and Whitin extended model with
backlogging costs, solved by Zangwill (1969) using dynamic programming
techniques. Modeling cooperation in purchasing and holding costs has already
appeared in literature. Additionally, our cooperation in backlogging is also
natural, but new: all the members of a coalition pay backlogging cost
(compensation for delayed demand) by the cheapest cost among those in the
coalition. The larger the coalition the stronger so that it can ``squeeze''
a bit more their customers. SI-games are an extension of PI-games (Guardiola
et al., 2009) since the latter do not include setup costs. The reader may
note that whenever set up costs are zero in all periods SI-games become
PI-games. Moreover, SI-games also generalize ELS-games in that all
considered costs can be different for the different players in each period
and, in addition, backorders are permitted.

The contribution of this paper is to prove that the above mentioned mode of
cooperation is always stable in that there exist fair allocations of the
overall cost of the system among the members of a coalition so that no
subgroup of agents is better off by leaving the consortium (every SI-game
has a nonempty core). Then we propose a parametric family of cost
allocations for SI-games: the extended Owen points. We provide sufficient
conditions for this to be a stable family against coalitional defections of
firms, that is, every extended Owen point is a core-allocation. Finally, we
focus on those periods of the time horizon that are consolidated and analyze
their effect on the stability of cost allocations. Specifically we prove
that for consolidated SI-situatuins, the single extended Owen point belongs
to the core of the game. Our paper contributes as well to the emerging
literature on the analysis of problems in Operations Research by means of
cooperative games. The interested reader is referred to Borm et al. (2001)
for further details on operations research games (including inventory games).

\section{\label{PIP}Model}

We start by describing the basic form of lot sizing problems (see Johnson
and Montgomery (1974) for further details). We focus here on periodic review
inventory problems where a setup cost is incurred when placing an order,
which in turns makes the cost structure non-linear.

A setup-inventory problem (hereafter SI-problem) can be described as
follows. We consider $T$ periods, numbered from $1$ to $T$, where the demand
for a single product occurs in each of them. This demand can be satisfied by
purchasing or by own production, and can be done in three different periods:
(i) the same period, (ii) an earlier period (as inventory), (iii) a later
period (as backlogging). Every time an order is placed in a certain period,
a fixed cost must be paid. Therefore, the model includes purchasing,
inventory holding, backlogging and setup costs. We assume, without loss of
generality, the initial and terminal inventories are set to zero. The
objective is to find an optimal ordering plan, that is a feasible ordering
plan that minimizes the sum of setup, purchasing, inventory holding and
backlogging cost.

For each period $t=1,\ldots ,T$ we let:

\begin{itemize}
\item $d_{t}=$ demand in period $t$ and $d=(d_{1},\ldots ,d_{T}).$

\item $k_{t}=$ setup cost in period $t$ and $k=(k_{1},\ldots ,k_{T}).$

\item $h_{t}=$ unit inventory carrying costs in period $t\ $and $%
h=(h_{1},\ldots ,h_{T}).$

\item $b_{t}=$ unit backlogging carrying costs in period $t\ $and $%
b=(b_{1},\ldots ,b_{T}).$

\item $p_{t}=$ unit purchasing costs in period $t\ $and $p=(p_{1},\ldots
,p_{T}).$

\item $q_{t}=$ order size in period $t.$

\item $I_{t}=$ ending inventory in period $t$.

\item $E_{t}=$ ending backlogged demand in period $t$.
\end{itemize}

We consider that costs and demand can never be negative. Furthermore, the
decision variables $q_{t}$, $I_{t}$ and $E_{t}$ take integer values. The
single-item formulation of the problem $(T,d,k,h,b,p)$ is as follows: 
\begin{eqnarray}
(SI)\quad &\min &\displaystyle\sum_{t=1}^{T}\left[
p_{t}q_{t}+h_{t}I_{t}+b_{t}E_{t}+k_{t}\gamma (q_{t})\right]  \notag \\
&\mbox{s.t.}&I_{0}=I_{T}=E_{0}=E_{T}=0,  \notag \\
&&I_{t}-E_{t}=I_{t-1}-E_{t-1}+q_{t}-d_{t},\quad t=1,\ldots ,T,  \notag \\
&&q_{t},\;I_{t},\;E_{t},\mbox{ non-negative, integer},\;t=1\ldots ,T,  \notag
\end{eqnarray}

where,%
\begin{equation*}
\gamma (q)=\left\{ 
\begin{array}{cc}
1 & \text{if }q>0, \\ 
0 & \text{if }q=0.%
\end{array}%
\right.
\end{equation*}

We define a feasible ordering plan for a SI-problem as $\sigma \in \mathbb{R}%
^{T}$ where $\sigma _{t}\in T\cup \{0\}$ denotes the period where demand of
period $t$ is ordered. We assume the convention that $\sigma _{t}=0$ if and
only if $d_{t}=0$. It means that there is no order placed to satisfy a null
demand at period $t$. Moreover, we define $P(\sigma )\in \mathbb{R}^{T}$ as
the cost vector associated to ordering plan $\sigma $ (henceforth: cost-plan
vector), where 
\begin{equation*}
P_{t}(\sigma )=\left\{ 
\begin{array}{ccc}
0 & \text{if} & \sigma _{t}=0, \\ 
p_{t}, & \text{if} & \sigma _{t}=t, \\ 
p_{\sigma _{t}}+\sum_{r=\sigma _{t}}^{t-1}h_{r}, & \text{if} & 1\leq \sigma
_{t}<t, \\ 
p_{\sigma _{t}}+\sum_{r=t+1}^{\sigma _{t}}b_{r}, & \text{if} & t<\sigma
_{t}\leq T.%
\end{array}%
\right.
\end{equation*}

If $\sigma ^{\ast }$ is an optimal ordering plan for a SI-problem, then the
optimal cost is given by 
\begin{equation*}
v(SI)=P(\sigma ^{\ast })^{\prime }d+\delta (\sigma ^{\ast })^{\prime
}k=\sum_{t=1}^{T}\left( P_{t}(\sigma ^{\ast })d_{t}+\delta _{t}(\sigma
^{\ast })k_{t}\right) ,
\end{equation*}

\noindent where, $\delta (\sigma ^{\ast })=\left( \delta _{t}(\sigma ^{\ast
})\right) _{t\in T}$ and 
\begin{equation*}
\delta _{t}(\sigma ^{\ast })=\left\{ 
\begin{array}{cc}
1 & \text{if }\exists r\in T/\sigma _{r}^{\ast }=t\text{ }, \\ 
0 & \text{otherwise.}%
\end{array}%
\right.
\end{equation*}

Notice that if all setup costs are zero, the problem we are dealing with is
a PI-problem (see Guardiola et al. 2009).

\section{\label{PIG}Cooperation in lot-sizing with heterogeneous costs}

Next we address a variant of this model where several firms, facing each one
a SI-problem, coordinate their actions to reduce costs. This coordination is
driven by sharing ordering channels, backlogged and inventory carrying
technologies. This means that cooperating firms make a joint order and pay
backlogged and inventory carrying demand at the cheapest costs among the
members of the coalition at each period. Formally, a setup-inventory
situation (henceforth, SI-situation) is a tuple $(N,D,K,H,B,P)$ where:

\begin{itemize}
\item $N=\{1,\ldots ,n\}$ is the set of players.

\item $D=[d^{1},\ldots ,d^{n}]^{\prime }$ is an integer demand matrix, where
each row corresponds to the demand of a player, that is, $%
d^{i}=[d_{1}^{i},\ldots ,d_{T}^{i}]^{\prime }.$

\item $K=[k^{1},\ldots ,k^{n}]^{\prime }$ is a setup cost matrix, where $%
k^{i}=[k_{1}^{i},\ldots ,k_{T}^{i}]^{\prime }.$

\item $H=[h^{1},\ldots ,h^{n}]^{\prime }$ is an inventory carrying costs
matrix, where $h^{i}=[h_{1}^{i},\ldots ,h_{T}^{i}]^{\prime }.$

\item $B=[b^{1},\ldots ,b^{n}]^{\prime }$ is a backlogging carrying costs
matrix, where $b^{i}=[b_{1}^{i},\ldots ,b_{T}^{i}]^{\prime }.$

\item $P=[p^{1},\ldots ,p^{n}]^{\prime }$ is a purchasing costs matrix,
where $p^{i}=[p_{1}^{i},\ldots ,p_{T}^{i}]^{\prime }.$\medskip
\end{itemize}

In order to simplify the notation we define $Z$ as a matrix in which all
costs are included, that is, $Z:=(K,H,B,P).$ A cost TU-game is a pair $(N,c)$%
, where $N$ is the finite player set and $c:\mathcal{P}(N)\rightarrow 
\mathbb{R}$ the characteristic function satisfying $c(\varnothing )=0.$ The
subgame related to coalition $S,c_{S},$ is the restriction of the mapping $c$
to the subcoalitions of $S.$ A cost allocation will be $x\in \mathbb{R}^{n}$
and, for every coalition $S\subseteq N$ we shall write $x(S):=\sum_{i\in
S}x_{i}$ the cost-sharing to coalition $S$ (where $x(\varnothing )=0).$

For each SI-situation $(N,D,Z)$ we associate a cost TU-game $(N,c)$ where,
for any nonempty coalition $S\subseteq N, c(S)$ is the optimal value of the
optimization problem $SI(S)$, defined as: 
\begin{eqnarray*}
(SI(S))\quad &\min
&\sum_{t=1}^{T}(p_{t}^{S}q_{t}+h_{t}^{S}I_{t}+b_{t}^{S}E_{t}+k_{t}^{S}\gamma
(q_{t})) \\
&\mbox{s.t.}&I_{0}=I_{T}=E_{0}=E_{T}=0 \\
&&I_{t}-E_{t}=I_{t-1}-E_{t-1}+q_{t}-d_{t}^{S},\quad t=1,\ldots ,T, \\
&&q_{t}\geq 0,\;I_{t}\geq 0,\;E_{t}\geq 0,\mbox{ and integers, }t=1,\ldots
,T;
\end{eqnarray*}%
with%
\begin{equation*}
p_{t}^{S}=\min_{i\in S}\{p_{t}^{i}\},\;h_{t}^{S}=\min_{i\in
S}\{h_{t}^{i}\},\;b_{t}^{S}=\min_{i\in
S}\{b_{t}^{i}\},\;k_{t}^{S}=\min_{i\in
S}\{k_{t}^{i}\},\;d_{t}^{S}=\sum_{i\in S}d_{t}^{i}.
\end{equation*}%
Notice that for all nonempty $S\subseteq N$ the characteristic function $c$
can be rewritten as follows:%
\begin{equation*}
c(S)=P^{S}(\sigma ^{S})^{\prime }d^{S}+\delta (\sigma ^{S})^{\prime
}k^{S}=\sum_{t=1}^{T}\left( P_{t}^{S}(\sigma ^{S})d_{t}^{S}+\delta
_{t}(\sigma ^{S})k_{t}^{S}\right) ,
\end{equation*}%
where $\sigma ^{S},P_{t}^{S}(\sigma ^{S})\in \mathbb{R}^{T}$ are the optimal
ordering plan and the cost-plan vector associated to $SI(S)$, respectively.
Every cost TU-game defined in this way is what we call a setup-inventory
game (SI-game).

The reader may notice that every PI-game (as introduced by Guardiola et al.,
2009) is a SI-game with $K=0.$ Hence the class of PI-games is a particular
subclass of the SI-games. \medskip

First, we wonder whether the above model of cooperation is stable, i.e.
whether there is a fair division of the total cost among the players such
that no group of them has incentives to leave. As we had already announced
the concept of core provides a direct answer to that question. The core of
the game $(N,c)$ consists of those cost allocations which divide the cost of
the grand coalition in such a way that any other coalition pays at most its
cost by the characteristic function. Formally, \medskip

$Core(N,c)=\left\{ x\in \mathbb{R}^{n}\left/ x(N)=c(N)\text{ and }x(S)\leq
c(S)\text{ for all }S\subset N\right. \right\}$. \medskip

In the following, fair allocations of the total cost will be called
core-allocations.

Bondareva (1963) and Shapley (1967) independently provide a general
characterization of games with a non-empty core by means of balanceness. A
collection of coalitions in $N$, $\mathcal{B}\subseteq \mathcal{P}(N)$ is a
balanced collection if there exist nonnegative weights $\left\{ \lambda
_{S}\right\} _{S\in \mathcal{B}}$ such that $\sum_{S\in \mathcal{B}:i\in
S}\lambda _{S}=1$ for all $i\in N$. Those weight $\left\{ \lambda
_{S}\right\} _{S\in \mathcal{B}}$ are called balancing coefficient. A cost
game $(N,c)$ is balanced if for every balanced collection $\mathcal{B}$ with
balancing coefficients $\{\lambda _{S}\}_{S\in \mathcal{B}}$ it holds that 
\begin{equation*}
\sum_{S\in \mathcal{B}}\lambda _{S}c(S)\geq c(N).
\end{equation*}%
Then, Bondareva and Shapley prove that $(N,c)$ has a nonempty core if and
only if it is balanced. In addition, it is totally balanced game if the core
of every subgame is nonempty. Totally balanced games were introduced by
Shapley and Shubik in the study of market games (see Shapley and Shubik,
1969).

It is important to remark that even though the problems that define SI-games
have totally unimodular constraint matrices, dual solutions do not induce
core solution throughout the Owen construction. The reason is because the
objective function is not linear. Therefore, the question whether the core
of the corresponding game is empty or not is a challenging query and its
study makes sense. \medskip

The main result of this section states that the cooperation in lot sizing
problems with heterogeneous costs is always stable. In other words, SI-games
are balanced. In what follows we include a technical lemma that helps in
proving the following theorem. Actually, it provides a procedure for
constructing new ordering plans out of existing ones. The rationale behind
this construction is similar to the one used in Van den Heuvel et al. (2007).

\begin{lemma}
\label{lem:prev} Let $\mathcal{B}$ be a balanced collection of coalitions
with balancing coefficients $\{\lambda _{S}\}_{S\in \mathcal{B}}$. Assume
that for each coalition $S\in \mathcal{B}$, $\pi ^{S}$ is the optimal order
plan for the problem $(T,d^{S},k^{N},h^{N},b^{N},p^{N})$. Let $r$ be the
smallest positive integer such that $l_{S}:=r\lambda _{S}^{\ast }\in \mathbb{%
Z}^{+}$ for all $S\in \mathcal{B}$ where $\{\lambda _{S}^{\ast }\}_{S\in 
\mathcal{B}}$ are rational numbers greater than or equal to $\{\lambda
_{S}\}_{S\in \mathcal{B}}$. Then, there exist $\{\psi ^{j}\}_{j=1\ldots ,r}$
feasible order plans for the problem $(T,d^{N},k^{N},h^{N},b^{N},p^{N})$
satisfying (i), (ii) and (iii):

\begin{itemize}
\item[(i)] For each $t$, there exists $S(j)\in \mathcal{B}$ such that $\psi
_{t}^{j}=\left\{ 
\begin{array}{ll}
\pi _{t}^{S(j)} & \text{if}\;d_{t}^{N}>0, \\ 
0 & \text{otherwise.}%
\end{array}%
\right. $

\item[(ii)] $S(j)$ is used at most $r\lambda _{S(j)}^{\ast }$ times $\forall
j=1,\ldots ,r$.

\item[(iii)] For any $j=1,\ldots ,r,\;P_{t}^{N}(\pi ^{S(j)})\leq
P_{t}^{N}(\pi ^{S})$ for all $S\neq S(i),\;\forall i=1,\ldots ,r$.
\end{itemize}
\end{lemma}

\begin{proof}
The proof of this lemma is constructive so that at the end we will have a
procedure to actually construct the corresponding order policies $\{\psi
^{j}\}_{j=1\ldots ,r}$.

Let $\mathcal{B}$ be a balanced collection and $\{\lambda _{S}\}_{S\in 
\mathcal{B}}$ their corresponding balancing coefficients. Take $\lambda
_{S}^{\ast }$ as a rational number greater than or equal to $\lambda _{S}$
for all $S\in \mathcal{B}$. There exists $r\in \mathbb{N}$ being the
smallest positive integer such that $\l _{S}=r\lambda _{S}^{\ast }\in 
\mathbb{Z}^{+}$ for all $S\in \mathcal{B}$. Notice that this number $r$
satisfies 
\begin{equation}
m:=\sum_{S\in \mathcal{B}}r\lambda _{S}^{\ast }\geq r,  \label{eq:def_m}
\end{equation}%
since $\sum_{S\in \mathcal{B}}\lambda _{S}^{\ast }\geq \sum_{S\in \mathcal{B}%
}\lambda _{S}\geq 1$.

Consider the following artificial set of coalitions, namely $\mathcal{BA}$%
\label{pag:BA}. For each coalition $S\in \mathcal{B}$, consider $%
l_{S}=r\lambda _{S}^{\ast }$ replicas of $S$ in the set $\mathcal{BA}$.
Therefore, we have in total $m$ coalitions in the new collection (see (\ref%
{eq:def_m})). We assume an arbitrary ordering of these coalitions so that we
can refer, without confusion, to any coalition, say $S_{k}$, by its index in
such a sequence. Moreover, let $\pi ^{S_{k}}$ be the optimal order plan
corresponding to the problem $(T,d^{S_{k}},k^{N},h^{N},b^{N},p^{N})$, for
each $k=1,\ldots ,m$.

Then we proceed to construct the feasible order plans $\{\psi^j\}$. For each
period $t$, $1\leq t\leq T$ set:

\begin{description}
\item[i)] $\psi _{t}^{j}=0$ for all $j=1,\ldots ,r$ if $d_{t}^{N}=0$.

\item[ii)] If $d_{t}^{N}>0$ do the following. Define 
\begin{eqnarray}
N_{t}^{\ast } &=&\{i\in N:d_{t}^{i}>0\},  \notag \\
\mathcal{C}_{t} &=&\{S_{k}:N_{t}^{\ast }\cap S_{k}\neq \emptyset
,\;k=1,\ldots ,m\}.  \label{d*}
\end{eqnarray}%
Notice that $N_{t}^{\ast }\neq \emptyset $ and moreover $\{i\}\subset
N_{t}^{\ast }$ since there exists always at least one agent $i\in N$ such
that $d_{t}^{i}>0$, otherwise $d_{t}^{N}=0$. Therefore, $\displaystyle%
\sum_{S\in \mathcal{B},i\in S}r\lambda _{S}^{\ast }\geq r$; and thus $|%
\mathcal{C}_{t}|\geq r$.

Arrange the coalitions in $\mathcal{C}_{t}$ in non-decreasing sequence , $%
S_{(1)},S_{(2)},\ldots ,S_{(|\mathcal{C}_{t}|)}$, with respect to the values
of $\{P_{t}^{N}(\pi ^{S_{k}})\}_{S^{k}\in \mathcal{C}_{t}}$. It is then
clear that: 
\begin{equation}
P_{t}^{N}(\pi ^{S_{(1)}})\leq P_{t}^{N}(\pi ^{S_{(2)}})\leq \ldots \leq
P_{t}^{N}(\pi ^{S_{(|\mathcal{C}_{t}|)}}).
\end{equation}%
Finally, we set $\psi _{t}^{j}=\pi _{t}^{S_{(j)}}$ for $j=1,\ldots ,r$.
Notice that since $|C_t|\ge r$ this definition is well-stated.
\end{description}

\noindent The above construction satisfies the thesis of the lemma.
\end{proof}

\medskip

The following example illustrates how to implement the aforementioned
procedure.

\begin{example}
Consider the following SI-situation with three players and four periods.
Notice that all players have the same costs but different demands. 
\begin{equation*}
\begin{array}{|c|c|c|c|c||c|c|c|c||c|c|c|c||c|c|c|c||c|c|c|c|}
\hline
& \multicolumn{4}{|c||}{Demand} & \multicolumn{4}{|c||}{Purchasing} & 
\multicolumn{4}{|c||}{Inventory} & \multicolumn{4}{|c||}{Backlogging} & 
\multicolumn{4}{|c|}{Setup} \\ \hline
P1 & $1$ & $1$ & $3$ & $2$ & $1$ & $1$ & $2$ & $2$ & $1$ & $2$ & $2$ & $1$ & 
$1$ & $1$ & $1$ & $1$ & $3$ & $1$ & $1$ & $9$ \\ \hline
P2 & $2$ & $1$ & $8$ & $2$ & $1$ & $1$ & $2$ & $2$ & $1$ & $2$ & $2$ & $1$ & 
$1$ & $1$ & $1$ & $1$ & $3$ & $1$ & $1$ & $9$ \\ \hline
P3 & $2$ & $1$ & $9$ & $2$ & $1$ & $1$ & $2$ & $2$ & $1$ & $2$ & $2$ & $1$ & 
$1$ & $1$ & $1$ & $1$ & $3$ & $1$ & $1$ & $9$ \\ \hline
\end{array}%
\end{equation*}
Let $B:=\left\{ \{1,2\},\{1,3\},\{2,3\},\{1,2,3\}\right\} $\ be a balanced
collection of coalitions with balancing coefficients $\{\lambda _{S}\}_{S\in 
\mathcal{B}}=\left\{ \frac{1}{3},\frac{1}{3},\frac{1}{3},\frac{1}{3}\right\}
.$\ In this case $\{\lambda _{S}^{\ast }\}_{S\in \mathcal{B}}=\{\lambda
_{S}\}_{S\in \mathcal{B}}$\ and $r=3.$

Next table shows the optimal ordering plans and the corresponding cost plan
vectors for the problem $(T,d^{S},k^{N},h^{N},b^{N},p^{N})$ for each
coalition $S\in \mathcal{B}$. 
\begin{equation*}
\begin{array}{|c|c|c|c|c||c|c|c|c||}
\hline
S & \pi _{1}^{S} & \pi _{2}^{S} & \pi _{3}^{S} & \pi _{4}^{S} & 
P_{1}^{N}(\pi ^{S}) & P_{2}^{N}(\pi ^{S}) & P_{3}^{N}(\pi ^{S}) & 
P_{4}^{N}(\pi ^{S}) \\ \hline
\{1,2\} & 1 & 2 & 3 & 3 & 1 & 1 & 2 & 4 \\ \hline
\{1,3\} & 2 & 2 & 3 & 3 & 2 & 1 & 2 & 4 \\ \hline
\{2,3\} & 1 & 2 & 3 & 3 & 1 & 1 & 2 & 4 \\ \hline
\{1,2,3\} & 1 & 2 & 3 & 4 & 1 & 1 & 2 & 2 \\ \hline
\end{array}%
\end{equation*}

The reader may notice that there always exist 3 feasible order plans for the
problem $(T,d^{N},k^{N},h^{N},b^{N},p^{N})$ given by

\begin{equation*}
\begin{array}{|c|c|c|c|c||c|c|c|c||}
\hline
j & \psi _{1}^{j} & \psi _{2}^{j} & \psi _{3}^{j} & \psi _{4}^{j} & 
P_{1}^{N}(\psi ^{j}) & P_{2}^{N}(\psi ^{j}) & P_{3}^{N}(\psi ^{j}) & 
P_{4}^{N}(\psi ^{j}) \\ \hline
1 & 1 & 2 & 3 & 4 & 1 & 1 & 2 & 2 \\ \hline
2 & 1 & 2 & 3 & 3 & 1 & 1 & 2 & 4 \\ \hline
3 & 1 & 2 & 3 & 3 & 1 & 1 & 2 & 4 \\ \hline
\end{array}%
\end{equation*}%
These plans can be built by means of the following recursive procedure: for
all $t=1,2,3,4$

\begin{itemize}
\item $\psi _{t}^{1}=\pi _{t}^{S(1)}$ such that $S(1)=\arg \min \left\{
P_{t}^{N}(\pi ^{S})\right\} ,$

\item $\psi _{t}^{2}=\pi _{t}^{S(2)}$ such that $S(2)=\arg \min \left\{
P_{t}^{N}(\pi ^{S})\left\vert S\neq S(1)\right. \right\} ,$

\item $\psi _{t}^{3}=\pi _{t}^{S(3)}$ such that $S(3)=\arg \min \left\{
P_{t}^{N}(\pi ^{S})\left\vert S\neq S(1),S(2)\right. \right\} ,$
\end{itemize}

In addition, for all $S\neq S(1),S(2),S(3),$ $P_{t}^{N}(\pi ^{S})\geq
P_{t}^{N}(\pi ^{S(j)})=P_{t}^{N}(\psi ^{j})$, for all $t=1,2,3,4$ and all $%
j=1,2,3$.
\end{example}

\begin{theorem}
\label{totalbal} Every SI-game $(N,c)$ associated to a SI-situation $(N,D,Z)$
is balanced.
\end{theorem}

\begin{proof}
Let $\mathcal{B}\subset 2^{N}$ be a balanced collection and $\left\{ \lambda
_{S}\right\} _{S\in \mathcal{B}}$ their corresponding balancing
coefficients. Then, 
\begin{eqnarray*}
\sum_{S\in \mathcal{B}}\lambda _{S}c(S) &=&\sum_{S\in \mathcal{B}}\lambda
_{S}\left( P^{S}(\sigma ^{S})^{\prime }d^{S}+\delta (\sigma ^{S})^{\prime
}k^{S}\right) \\
&\geq &\sum_{S\in \mathcal{B}}\lambda _{S}\left( P^{N}(\sigma ^{S})^{\prime
}d^{S}+\delta (\sigma ^{S})^{\prime }k^{N}\right) \\
&\geq &\sum_{S\in \mathcal{B}}\lambda _{S}\left( P^{N}(\pi ^{S})^{\prime
}d^{S}+\delta (\pi ^{S})^{\prime }k^{N}\right) ,
\end{eqnarray*}

\noindent where $\pi ^{S}$ for each $S\in B$\ is an optimal order plan for
the SI-problem $(T,d^{S},k^{N},h^{N},b^{N},p^{N})$.

Our goal is to prove that 
\begin{equation}
\sum_{S\in \mathcal{B}}\lambda _{S}\left( P^{N}(\pi ^{S})^{\prime
}d^{S}+\delta (\pi ^{S})^{\prime }k^{N}\right) \geq c(N).  \label{Objetivo}
\end{equation}

According to the construction in Lemma \ref{lem:prev}, take $\lambda
_{S}^{\varepsilon ^{S}}$ as a rational number greater than or equal to $%
\lambda _{S}$ such that 
\begin{equation}
\left\{ 
\begin{array}{ll}
\lambda _{S}^{\varepsilon ^{S}}-\lambda _{S}=\epsilon ^{S}=0 & \mbox{if }%
\lambda _{S}\mbox{ is rational} \\ 
\lambda _{S}^{\varepsilon ^{S}}-\lambda _{S}=\epsilon ^{S}>0 & %
\mbox{otherwise}.%
\end{array}%
\right.  \label{rac_con}
\end{equation}%
There exists $r\in N$ being the smallest positive integer for which $%
r\lambda _{S}^{\varepsilon ^{S}}$ is integral for all $S\in B$. Then (\ref%
{Objetivo}) can be rewritten as 
\begin{equation}
\sum_{S\in \mathcal{B}}r\lambda _{S}^{\varepsilon ^{S}}\left( P^{N}(\pi
^{S})^{\prime }d^{S}+\delta (\pi ^{S})^{\prime }k^{N}\right) -\sum_{S\in 
\mathcal{B}}r\varepsilon ^{S}\left( P^{N}(\pi ^{S})^{\prime }d^{S}+\delta
(\pi ^{S})^{\prime }k^{N}\right) \geq rc(N).  \label{Objetivo 2}
\end{equation}

Using the set $BA$ defined in page \pageref{pag:BA}, the first term in the
left-hand side of (\ref{Objetivo 2}) can be rewritten as 
\begin{eqnarray}
\hspace*{-0.4cm}\displaystyle\sum_{k=1}^{m}\left( P^{N}(\pi
^{S_{k}})^{\prime }d^{S_{k}}+\delta (\pi ^{S_{k}})^{\prime }k^{N}\right) &&%
\displaystyle\hspace*{-0.4cm}=\sum_{k=1}^{m}\sum_{t=1}^{T}\left(
P_{t}^{N}(\pi ^{S_{k}})d_{t}^{S_{k}}+\delta _{t}(\pi
^{S_{k}})k_{t}^{N}\right)  \label{Continuacion} \\
&&\displaystyle\hspace*{-3cm}=\sum_{t=1}^{T}\sum_{i=1}^{n}d_{t}^{i}\sum 
_{\substack{ \{k: S_k\backepsilon i\}}}P_{t}^{N}(\pi
^{S_{k}})+\sum_{t=1}^{T}\sum_{k=1}^{m}\delta _{t}(\pi ^{S_{k}})k_{t}^{N}.
\label{chorizo}
\end{eqnarray}

Now, the expression in (\ref{chorizo}) is stated as: 
\begin{equation}
\hspace*{-1cm}\sum_{t=1}^{T}\Big(\sum_{\substack{ i=1  \\ d_{t}^{i}=0}}%
^{n}d_{t}^{i}\sum_{\substack{ \{k: S_k\backepsilon i\}}}P_{t}^{N}(\pi
^{S_{k}})+\sum_{\substack{ i=1  \\ d_{t}^{i}>0}}^{n}d_{t}^{i}\sum_{\substack{
\{k: S_k\backepsilon i\}}}P_{t}^{N}(\pi ^{S_{k}})\Big)+\sum_{t=1}^{T}%
\sum_{k=1}^{m}\delta _{t}(\pi ^{S_{k}})k_{t}^{N}.  \label{eq:chori1}
\end{equation}

From the above formula and using the definition in (\ref{d*}) we get: 
\begin{equation}
\hspace*{-1cm}\sum_{t=1}^{T}\sum_{\substack{ i=1  \\ d_{t}^{i}>0}}%
^{n}d_{t}^{i}\sum_{\substack{ \{k: S_k\backepsilon i\}}}P_{t}^{N}(\pi
^{S_{k}})+\sum_{t=1}^{T}\sum_{k=1}^{m}\delta _{t}(\pi ^{S_{k}})k_{t}^{N}\geq
\sum_{t=1}^{T}\sum_{i=1}^{n}d_{t}^{i}\sum_{j=1}^{r}P_{t}^{N}(\psi
^{j})+\sum_{t=1}^{T}\sum_{j=1}^{r}\delta _{t}(\psi ^{j})k_{t}^{N}
\label{eq:chori2}
\end{equation}

\noindent Hence, since $\sum_{i=1}^{n}d_{t}^{i}=d_{t}^{N}$, the right-hand
side of (\ref{eq:chori2}) equals the following: 
\begin{eqnarray}
&=&\sum_{j=1}^{r}\sum_{t=1}^{T}\left( P_{t}^{N}(\psi ^{j})d_{t}^{N}+\delta
_{t}(\psi ^{j})k_{t}^{N}\right)  \notag \\
&\geq &\sum_{j=1}^{r}\sum_{t=1}^{T}\left( P_{t}^{N}(\sigma
^{N})d_{t}^{N}+\delta _{t}(\sigma ^{N})k_{t}^{N}\right)  \notag \\
&=&\sum_{j=1}^{r}c(N)=rc(N).  \label{in_fin}
\end{eqnarray}%
Thus, finally from (\ref{Continuacion}) and (\ref{in_fin}) we get the
following inequality: 
\begin{equation}
\sum_{S\in \mathcal{B}}r\lambda _{S}^{\varepsilon ^{S}}\left( P^{N}(\pi
^{S})^{\prime }d^{S}+\delta (\pi ^{S})^{\prime }k^{N}\right) \geq
rc(N),\;\forall \lambda _{S}^{\varepsilon ^{S}}\mbox{ satisfying }(\ref%
{rac_con}).  \label{eq_fin}
\end{equation}

Hence, taking limit in (\ref{eq_fin}) when $\lambda _{S}^{\varepsilon
^{S}}\rightarrow \lambda _{S}$ for all $S\in B$ we obtain: 
\begin{equation*}
\sum_{S\in \mathcal{B}}\lambda _{S}\left( P^{N}(\pi ^{S})^{\prime
}d^{S}+\delta (\pi ^{S})^{\prime }k^{N}\right) \geq c(N),
\end{equation*}%
what concludes the proof.
\end{proof}

\medskip

We note in passing that every subgame of a SI-game is a new SI-game. Thus,
Theorem \ref{totalbal} implies that every SI-game is totally balanced.

There exists an alternative proof of the balancedness of this class of
games. Here we outline this proof for the sake of completeness.

From the central part of the proof of Theorem \ref{totalbal} we deduce that
the balanced character is ensured for any balanced collection with rational
balancing coefficients. Notice that balanced coefficients must be optimal
solutions to the following linear problem (for any suitable choice of
coefficients): 
\begin{eqnarray*}
\max &&\sum_{S\subset N}\lambda _{S}c(S)\mbox{ \hspace*{2cm}} \\
\mbox{s.t.} &&\sum_{S:S\backepsilon i}\lambda _{S}=1\quad i=1,\ldots ,n \\
&&\lambda _{S}\geq 0\quad \forall S\subset N.
\end{eqnarray*}%
Therefore, since the feasible region of the above problem has all its
extreme points being rational numbers we deduce that balancedness holds for
those choices. In addition, any non rational family of balanced coefficients
must be a convex combination of extreme points in this polyhedron. Hence, we
can apply the following construction.

Let $\mathcal{B}$ be a balanced collection with non-rational balancing
coefficients $\{\lambda^{B}_S\}_{S\in \mathcal{B}}$. There exist $%
B^1,\ldots,B^k$ balanced collections with rational balancing coefficients $%
\Big\{\{\lambda^{B^1}_S\}_{S\in {B^1}},\ldots, \{\lambda^{B^k}_S\}_{S\in {B^k%
}}\Big\}$ and $\alpha=(\alpha^1,\ldots,\alpha^k)\ge 0$, $\sum_{i=1}^k
\alpha^i=1$ such that $\mathcal{B}=\bigcup_{i=1}^k B^i$ and $%
\lambda_S^B=\sum_{i=1}^k \alpha^i \lambda_S^{B^i}$. (We assume that $%
\lambda_S^{B^i}=0$ whenever $S\not \in B^i$.) Finally, 
\begin{equation*}
\sum_{S\in \mathcal{B}} \lambda^B_S c(S)=\sum_{S\in \mathcal{B}}
\sum_{i=1}^k \alpha^i \lambda_S^{B^i} c(S)= \sum_{i=1}^k \alpha^i \Big(%
\sum_{S\in B^i} \lambda^{B^i}_S c(S)\Big) \ge c(N).
\end{equation*}

In fact, from the above argument we deduce something more general: proving
balancedness for collections with rational balancing coefficients suffices.

\section{Extended Owen points}

We have just proven the stability of the grand coalition, in the sense of
the core. We know that there always exists a core-allocation for SI-games
but we do not know how to construct it. We propose to find suitable
cost-allocations for SI-games which are easy to calculate and satisfy good
properties.

The Owen point, introduced in Guardiola et al. (2009), is a core-allocation
for PI-games which represents the cost that each player has to pay when
producing at the minimum operational cost (see also Guardiola et al., 2008).
If we consider a SI-situation $(N,D,Z)$ with $K=0$, that is a PI-situation,
the Owen point, $o=(o_{i})_{i\in N}$, is given by 
\begin{equation*}
o_{i}=\sum_{t=1}^{T}P_{t}^{N}(\sigma ^{N})d_{t}^{i},\mbox{ for all
}i\in N.
\end{equation*}

In this section we introduce a parametric family of cost allocations with
the flavor of the Owen point but appropriate to SI-games. We call it the
family of extended Owen points. The interested reader is referred to Perea
et al. (2009, 2012) for alternative extensions of the concept of Owen point.
Before defining this new family of cost allocations, we need to introduce
some previous concepts. \medskip

Let $(N,D,Z)$ be a SI-situation and $(N,c)$ the associated SI-game. We
define the reduced SI-situation associated to $(N,D,Z)$ as a SI-situation $%
(N,D,\widetilde{Z})$ with $\widetilde{Z}=(\widetilde{K},\widetilde{H},%
\widetilde{B},\widetilde{P})$ where 
\begin{equation*}
\widetilde{K}=[k^{N},\ldots ,k^{N}]^{\prime },\widetilde{H}=[h^{N},\ldots
,h^{N}]^{\prime },\widetilde{B}=[b^{N},\ldots ,b^{N}]^{\prime },\widetilde{P}%
=[p^{N},\ldots ,p^{N}]^{\prime }.
\end{equation*}

Note that reduced SI-situations are the simplest SI-situations in that all
their costs are the same for all players in all periods. \medskip

We denote by $(N,\widetilde{c})$ the cost game associated to the reduced
SI-situation $(N,D,\widetilde{Z})$. Notice that $\widetilde{c}(S)\leq c(S)$
for all $S\subset N$ and $\widetilde{c}(N)=c(N).$ Hence $Core(N,\widetilde{c}%
)\subseteq Core(N,c).$ Clearly, each ELS-situation corresponds with a
reduced SI-situation for an appropriate choice of parameters since the costs
involved in each period are the same for all the players (see Van Den Heuvel
et al., 2007). Hence, ELS-situations are particular cases of SI-situations.
\medskip

Next we define the following sets:

\begin{itemize}
\item Set of ordering periods: $T^{S}:=\left\{ t\in T\left\vert \delta
_{t}(\sigma ^{S})=1\right. \right\} $ for all $S \subseteq N $.

It is easy to check that,

\begin{equation*}
\sum_{t=1}^{T}\delta _{t}(\sigma ^{S})k_{t}^{S}=\sum_{t\in T^{S}}k_{t}^{S}.
\end{equation*}

\item Set of consolidated periods: $\Upsilon :=\{t\in T|\exists i\in N%
\mbox{
such that }\delta _{t}(\sigma ^{S})=1$ $\mbox{ for all }S\subseteq N\text{
with }i\in S\}.$ A period is consolidated if there exists at least one
player such that he forces placing an order at this period to any coalition
that he belongs to.
\end{itemize}

We can distinguish two classes of costs for every coalition $S\subseteq N.$\
Variable costs $P^{S}(\sigma ^{S})^{\prime }d^{S},$\ which depends on
demands, and non-consolidated fixed costs $\sum_{t\in T^{S}\setminus
\Upsilon }k_{t}^{S}$. Next we define for each $S\subseteq N$%
\begin{eqnarray*}
N(S) &:&=P^{N}(\sigma ^{N})^{\prime }d^{S}-P^{S}(\sigma ^{S})d^{S}, \\
M(S) &:&=\sum_{t\in T^{N}\setminus \Upsilon }k_{t}^{N}-\sum_{t\in
T^{S}\setminus \Upsilon }k_{t}^{S}.
\end{eqnarray*}

Notice that $N(S)$ and $M(S)$ represent the difference between the ordering
plans $\sigma ^{S}$ and $\sigma ^{N}$ related to variable and
non-consolidated fixed costs.

We are ready now to define the family of extender Owen points.

\begin{equation}
\left\{ \omega (\alpha )\in \mathbb{R}^{N}:\alpha \in \mathbb{R}_{+}^{N}\ 
\text{such that }\alpha (N)>0\right\}  \label{eq:alfa}
\end{equation}%
where 
\begin{equation*}
\omega _{i}(\alpha ):=\sum_{t=1}^{T}P_{t}^{N}(\sigma
^{N})d_{t}^{i}+\sum_{t\in \Upsilon /i\in J_{t}}\frac{k_{t}^{N}}{\left\vert
J_{t}\right\vert }+\frac{\alpha _{i}}{\alpha (N)}\sum_{t\in T^{N}\setminus
\Upsilon }k_{t}^{N}
\end{equation*}%
for all $i\in N$ and $J_{t}:=\{i\in N$ such that $\delta _{t}(\sigma
^{i})=1\}.$

Notice that the above family of cost allocation is a parametric family
depending on $\alpha \in \mathbb{R}_{+}^{N}$ such that $\alpha (N)>0.$

Next proposition shows that, if the optimal ordering plan for the grand
coalition reduces variable and non-consolidated fixed costs with respect to
any coalition $S\subseteq N$, then the family of extended Owen points is a
core-allocation family.

\begin{proposition}
\label{alphas}Let $(N,D,Z)$ a SI-situation and $(N,c)$ the corresponding
SI-game. If $N(S),M(S)\leq 0$ for all $S\subseteq N,$ then for each $\alpha
\in \mathbb{R}_{+}^{N}\ $ such that $\alpha (N)>0,$ the allocation $\omega
(\alpha )=(\omega _{1}(\alpha ),\ldots ,\omega _{n}(\alpha ))$ defined in %
\eqref{eq:alfa} is a core-allocation.
\end{proposition}

\begin{proof}
Let $(N,D,Z)$ be a SI-situation and $(N,c)$\ the corresponding SI-game.
Then, 
\begin{eqnarray*}
\omega (S) &=&\sum_{t=1}^{T}P_{t}^{N}(\sigma ^{N})d_{t}^{S}+\sum_{i\in
S}\sum_{t\in \Upsilon /i\in J_{t}}\frac{k_{t}^{N}}{\left\vert
J_{t}\right\vert }+\frac{\alpha (S)}{\alpha (N)}\sum_{t\in T^{N}\setminus
\Upsilon }k_{t}^{N} \\
&\leq &\sum_{t=1}^{T}P_{t}^{S}(\sigma ^{S})d_{t}^{S}+\sum_{\substack{ t\in
\Upsilon  \\ \delta _{t}(\sigma ^{S})=1}}\left\vert J_{t}\right\vert \frac{%
k_{t}^{S}}{\left\vert J_{t}\right\vert }+\frac{\alpha (S)}{\alpha (N)}%
\sum_{t\in T^{S}\setminus \Upsilon }k_{t}^{S}+\left( N(S)+\frac{\alpha (S)}{%
\alpha (N)}M(S)\right) \\
&\leq &\sum_{t=1}^{T}P_{t}^{S}(\sigma ^{S})d_{t}^{S}+\sum_{t=1}^{T}\delta
_{t}(\sigma ^{S})k_{t}^{S}=c(S).
\end{eqnarray*}

It is easy to check that $\omega $ is efficient. Hence, $\omega \in
Core(N,c).$
\end{proof}

\medskip

The above result is illustrated in the next example.

\begin{example}
Consider the following SI-situation with three periods and three players: 
\begin{equation*}
\begin{array}{|c|c|c|c||c|c|c||c|c|c||c|c|c||c|c|c|}
\hline
& \multicolumn{3}{|c||}{Demand} & \multicolumn{3}{|c||}{Purchasing} & 
\multicolumn{3}{|c||}{Inventory} & \multicolumn{3}{|c||}{Backlogging} & 
\multicolumn{3}{|c|}{Setup} \\ \hline
P1 & $6$ & $5$ & $2$ & $3$ & $1$ & $1$ & $2$ & $3$ & $3$ & $1$ & $3$ & $1$ & 
$4$ & $3$ & $4$ \\ \hline
P2 & $4$ & $1$ & $1$ & $5$ & $1$ & $4$ & $2$ & $1$ & $3$ & $1$ & $3$ & $1$ & 
$0$ & $2$ & $5$ \\ \hline
P3 & $1$ & $4$ & $1$ & $2$ & $1$ & $3$ & $3$ & $1$ & $1$ & $1$ & $3$ & $1$ & 
$0$ & $0$ & $5$ \\ \hline
\end{array}%
\end{equation*}
The corresponding SI-game is given by:%
\begin{equation*}
\begin{array}{|c|c|c|c||c|c|c||c|c|c||c|c|c||c|c|c||c||}
\hline
S & d_{1}^{_{S}} & d_{2}^{_{S}} & d_{3}^{_{S}} & p_{1}^{_{S}} & p_{2}^{_{S}}
& p_{3}^{_{S}} & h_{1}^{_{S}} & h_{2}^{_{S}} & h_{3}^{_{S}} & b_{1}^{_{S}} & 
b_{2}^{_{S}} & b_{3}^{_{S}} & k_{1}^{_{S}} & k_{2}^{_{S}} & k_{3}^{_{S}} & 
c(S) \\ \hline
\{1\} & 6 & 5 & 2 & 3 & 1 & 1 & 2 & 3 & 3 & 1 & 3 & 1 & 4 & 3 & 4 & 35 \\ 
\hline
\{2\} & 4 & 1 & 1 & 5 & 1 & 4 & 2 & 1 & 3 & 1 & 3 & 1 & 0 & 2 & 5 & 24 \\ 
\hline
\{3\} & 1 & 4 & 1 & 2 & 1 & 3 & 3 & 1 & 1 & 1 & 3 & 1 & 0 & 0 & 5 & 8 \\ 
\hline
\{1,2\} & 10 & 6 & 3 & 3 & 1 & 1 & 2 & 1 & 3 & 1 & 3 & 1 & 0 & 2 & 4 & 47 \\ 
\hline
\{1,3\} & 7 & 9 & 3 & 2 & 1 & 1 & 2 & 1 & 1 & 1 & 3 & 1 & 0 & 0 & 4 & 29 \\ 
\hline
\{2,3\} & 5 & 5 & 2 & 2 & 1 & 3 & 2 & 1 & 1 & 1 & 3 & 1 & 0 & 0 & 5 & 22 \\ 
\hline
\{1,2,3\} & 11 & 10 & 4 & 2 & 1 & 1 & 2 & 1 & 1 & 1 & 3 & 1 & 0 & 0 & 4 & 43
\\ \hline
\end{array}%
\end{equation*}

Next table shows the optimal ordering plans, the corresponding cost-plan
vectors, and the differences between the ordering plans:%
\begin{equation*}
\begin{array}{|c|c|c|c||c|c|c||c||c||}
\hline
S & \sigma _{1}^{S} & \sigma _{2}^{S} & \sigma _{3}^{S} & P_{1}^{S}(\sigma
^{S}) & P_{2}^{S}(\sigma ^{S}) & P_{3}^{S}(\sigma ^{S}) & N(S) & M(S) \\ 
\hline
\{1\} & 1 & 2 & 3 & 3 & 1 & 1 & -6 & -4 \\ \hline
\{2\} & 2 & 2 & 2 & 4 & 1 & 2 & -9 & 0 \\ \hline
\{3\} & 1 & 2 & 2 & 2 & 1 & 2 & -1 & 0 \\ \hline
\{1,2\} & 1 & 2 & 2 & 3 & 1 & 2 & -13 & 0 \\ \hline
\{1,3\} & 1 & 2 & 2 & 2 & 1 & 2 & -3 & 0 \\ \hline
\{2,3\} & 1 & 2 & 2 & 2 & 1 & 2 & -2 & 0 \\ \hline
\{1,2,3\} & 1 & 2 & 2 & 2 & 1 & 2 & 0 & 0 \\ \hline
\end{array}%
\end{equation*}%
Since $N(S),M(S)\leq 0$\ for all $S\subseteq N,$\ we can conclude that

$\left\{ (19,13,7)+\frac{4}{\alpha (N)}(\alpha _{1},\alpha _{2},\alpha
_{3})\left\vert 
\begin{array}{c}
\alpha _{i}\in \mathbb{R}_{+}\text{ }\forall i\in N \\ 
\text{with }\alpha (N)>0%
\end{array}%
\right. \right\} \subseteq Core(N,c).$
\end{example}

We note that $\omega (\alpha )$\ is not a game-theoretical solution since
its definition only applies on SI-situations. A weaker sufficient condition
to ensure the cost allocation $\omega (\alpha )$\ to be in the core is given
by the following corollary.

\begin{corollary}
Let $(N,D,Z)$ a SI-situation and $(N,c)$ the corresponding SI-game. If there
exists $\alpha \in R_{+}^{N}\ $such that $\alpha (N)>0$ and $N(S)+\frac{%
\alpha (S)}{\alpha (N)}M(S)\leq 0$ for all $S\subseteq N,$ then $\omega
(\alpha )$ is a core-allocation.
\end{corollary}

The example below illustrates the above condition.

\begin{example}
Consider the SI-situation described by the following table.

\begin{equation*}
\begin{array}{|c|c|c|c||c|c|c||c|c|c||c|c|c||c|c|c|}
\hline
& \multicolumn{3}{|c||}{Demand} & \multicolumn{3}{|c||}{Purchasing} & 
\multicolumn{3}{|c||}{Inventory} & \multicolumn{3}{|c||}{Backlogging} & 
\multicolumn{3}{|c|}{Setup} \\ \hline
P1 & $5$ & $5$ & $2$ & $3$ & $1$ & $1$ & $2$ & $3$ & $3$ & $1$ & $2$ & $1$ & 
$3$ & $1$ & $2$ \\ \hline
P2 & $4$ & $1$ & $1$ & $5$ & $1$ & $4$ & $2$ & $1$ & $3$ & $1$ & $2$ & $1$ & 
$3$ & $1$ & $2$ \\ \hline
P3 & $1$ & $4$ & $5$ & $2$ & $1$ & $3$ & $3$ & $2$ & $1$ & $1$ & $2$ & $1$ & 
$2$ & $1$ & $2$ \\ \hline
\end{array}%
\end{equation*}
The corresponding SI-game is shown in the next table: 
\begin{equation*}
\begin{array}{|c|c|c|c||c|c|c||c|c|c||c|c|c||c|c|c||c||}
\hline
& d_{1}^{_{S}} & d_{2}^{_{S}} & d_{3}^{_{S}} & p_{1}^{_{S}} & p_{2}^{_{S}} & 
p_{3}^{_{S}} & h_{1}^{_{S}} & h_{2}^{_{S}} & h_{3}^{_{S}} & b_{1}^{_{S}} & 
b_{2}^{_{S}} & b_{3}^{_{S}} & k_{1}^{_{S}} & k_{2}^{_{S}} & k_{3}^{_{S}} & c
\\ \hline
\{1\} & 5 & 5 & 2 & 3 & 1 & 1 & 2 & 3 & 3 & 1 & 2 & 1 & 3 & 1 & 2 & 25 \\ 
\hline
\{2\} & 4 & 1 & 1 & 5 & 1 & 4 & 2 & 1 & 3 & 1 & 2 & 1 & 3 & 1 & 2 & 16 \\ 
\hline
\{3\} & 1 & 4 & 5 & 2 & 1 & 3 & 3 & 2 & 1 & 1 & 2 & 1 & 2 & 1 & 2 & 23 \\ 
\hline
\{1,2\} & 9 & 6 & 3 & 3 & 1 & 1 & 2 & 1 & 3 & 1 & 2 & 1 & 3 & 1 & 2 & 39 \\ 
\hline
\{1,3\} & 6 & 9 & 7 & 2 & 1 & 1 & 2 & 2 & 1 & 1 & 2 & 1 & 2 & 1 & 2 & 33 \\ 
\hline
\{2,3\} & 5 & 5 & 6 & 2 & 1 & 3 & 2 & 1 & 1 & 1 & 2 & 1 & 2 & 1 & 2 & 30 \\ 
\hline
\{1,2,3\} & 10 & 10 & 8 & 2 & 1 & 1 & 2 & 1 & 1 & 1 & 2 & 1 & 2 & 1 & 2 & 43
\\ \hline
\end{array}%
\end{equation*}

The optimal ordering plans, the corresponding cost-plan vectors, and the
differences between the ordering plans can be found in the last table:%
\begin{equation*}
\begin{array}{|c|c|c|c||c|c|c||c||c||}
\hline
& \sigma _{1}^{S} & \sigma _{2}^{S} & \sigma _{3}^{S} & P_{1}^{S}(\sigma
^{S}) & P_{2}^{S}(\sigma ^{S}) & P_{3}^{S}(\sigma ^{S}) & N(S) & M(S) \\ 
\hline
\{1\} & 2 & 2 & 3 & 3 & 1 & 1 & -5 & 2 \\ \hline
\{2\} & 2 & 2 & 2 & 3 & 1 & 2 & -5 & 2 \\ \hline
\{3\} & 2 & 2 & 2 & 3 & 1 & 3 & -11 & 2 \\ \hline
\{1,2\} & 2 & 2 & 3 & 3 & 1 & 1 & -9 & 2 \\ \hline
\{1,3\} & 1 & 2 & 3 & 2 & 1 & 1 & 0 & 0 \\ \hline
\{2,3\} & 1 & 2 & 2 & 2 & 1 & 2 & -6 & 0 \\ \hline
\{1,2,3\} & 1 & 2 & 3 & 2 & 1 & 1 & 0 & 0 \\ \hline
\end{array}%
\end{equation*}%
It can be easily checked that $N(S)+\frac{\alpha (S)}{\alpha (N)}M(S)\leq 0$%
\ for all $\alpha \in R_{+}^{N}\ $such that $\alpha (N)>0$\ and all $%
S\subseteq N.$ Hence%
\begin{equation*}
\left\{ \left( \frac{58}{3},\frac{31}{3},\frac{34}{3}\right) +\frac{2}{%
\alpha (N)}(\alpha _{1},\alpha _{2},\alpha _{3})\left\vert 
\begin{array}{c}
\alpha _{i}\in \mathbb{R}_{+}\text{ }\forall i\in N \\ 
\text{with }\alpha (N)>0%
\end{array}%
\right. \right\} \subseteq Core(N,c).
\end{equation*}
\end{example}

We finish this section with a simpler sufficient condition to check core
membership.

\begin{corollary}
Let $(N,D,Z)$ be a SI-situation and $(N,c)$ the corresponding SI-game. If
the reduced SI-situation $(N,D,\widetilde{Z})$ satisfies one of the
following conditions:

\begin{itemize}
\item[(i)] $\widetilde{N}(S),\widetilde{M}(S)\leq 0,$

\item[(ii)] there exists $\alpha \in R_{+}^{N}\ $such that $\alpha (N)>0$\
and $\widetilde{N}(S)+\frac{\alpha (S)}{\alpha (N)}\widetilde{M}(S)\leq 0$\
for all $S\subseteq N,$
\end{itemize}

where $\widetilde{N}(S),\widetilde{M}(S)$\ are the corresponding values for
the reduced SI-situation for all $S\subseteq N,$ then $\omega (\alpha )$\ is
a core-allocation.
\end{corollary}

\section{Consolidated situations and stability}

We focus now on those SI-situations that are consolidated. Then we analyze
their effect on the stability of the extended Owen points.

A consolidated SI-situation is described by means of a property of the
ordering periods: whenever a player places an order in a period every
coalition that contains that player places an order in the same period as
well. This idea of consolidation is a refinement of the original scheme
since it makes coalitions to perform as any of its individual members. The
concept is formalized in the following definition.

\begin{definition}
A SI-situation $(N,D,Z)$ is consolidated if $T^{S}\subseteq \Upsilon $ for
all $S\subseteq N$.
\end{definition}

From the above definition, it is clear that in any consolidated
SI-situation, $M(S)=0$ for any coalition $S\subseteq N$, since $\sum_{t\in
T^{S}\setminus \Upsilon }k_{t}^{S}=0$. Based in this fact, we can provide an
extended Owen point in the core for consolidated SI-games. The following
technical lemma is needed to prove this result.

\begin{lemma}
\label{CSIG} Let $(N,D,Z)$ be a consolidated SI-situation and $(N,c)$ the
corresponding SI-game. Then $P_{t}^{S}(\sigma ^{S})\geq P_{t}^{R}(\sigma
^{R})$ for all $t\in T$ with $d_{t}^{S}\neq 0$ and for all $S\subseteq
R\subseteq N.$
\end{lemma}

\begin{proof}
Suppose that $\exists t^{\prime }\in T$ with $d_{t^{\prime }}^{S}\neq 0$
such that $P_{t^{\prime }}^{S}(\sigma ^{S})<P_{t^{\prime }}^{R}(\sigma ^{R})$
then $\sigma _{t^{\prime }}^{S}=r$ and $\sigma _{t^{\prime }}^{R}=r^{\prime
} $ with $r\neq r^{\prime }.$ $\delta _{r}(\sigma ^{S})=1$ therefore $\delta
_{r}(\sigma ^{R})=1$ since $(N,D,Z)$ is a consolidated SI-situation. If we
take the next feasible plan, 
\begin{equation*}
\sigma ^{\ast }:=\left\{ 
\begin{array}{cc}
\sigma _{t}^{R} & \text{if }t\neq t^{\prime }, \\ 
&  \\ 
r & \text{if }t=t^{\prime },%
\end{array}%
\right.
\end{equation*}%
then $P_{t^{\prime }}^{R}(\sigma ^{\ast })\leq P_{t^{\prime }}^{S}(\sigma
^{S})<P_{t^{\prime }}^{R}(\sigma ^{R})$ and $\delta (\sigma ^{\ast
})^{\prime }k^{R}\leq \delta (\sigma ^{R})^{\prime }k^{R}.$ Hence, 
\begin{equation*}
c(R)=P^{R}(\sigma ^{R})^{\prime }d^{R}+\delta (\sigma ^{R})^{\prime
}k^{R}>P^{R}(\sigma ^{\ast })^{\prime }d^{R}+\delta (\sigma ^{\ast
})^{\prime }k^{R},
\end{equation*}%
and this is a contradiction because $\sigma ^{R}$ is an optimal ordering
plan of coalition $R\subseteq N$.
\end{proof}

Note that the above lemma exhibits a monotonicity property with respect to
the ordering policies. We mean that, the smaller the coalition, the greater
the cost of satisfying demand in each single period.

\begin{proposition}
Let $(N,D,Z)$ be a consolidated SI-situation and $(N,c)$ the asociated
SI-game. Then, the allocation $\psi \in \mathbb{R}^{N}$ given by 
\begin{equation*}
\psi _{i}:=\sum_{t=1}^{T}P_{t}^{N}(\sigma ^{N})d_{t}^{i}+\sum_{t\in \Upsilon
/i\in J_{t}}\frac{k_{t}^{N}}{\left\vert J_{t}\right\vert },
\end{equation*}%
for all $i\in N$ where $J_{t}:=\{i\in N$ such that $\delta _{t}(\sigma
^{i})=1\},$ is a core-allocation.
\end{proposition}

\begin{proof}
We suppose that $(N,D,Z)$ is consolidated. Note that if $T^{S}\subseteq
\Upsilon $ for $S\subseteq N$ then $M(S)=0$ and $\sum_{t\in T^{N}\setminus
\Upsilon }k_{t}^{N}=0.$ We know $\sigma _{t}^{S}\in \Upsilon $ for all $t\in
T$ and for all $S\subseteq N$ with $d_{t}^{S}\neq 0,$ then by Lemma \ref%
{CSIG}, $P^{S}(\sigma^{S})\geq P^{N}(\sigma ^{N}).$ Hence $N(S)\leq 0$ for
all $S\subseteq N$, By Proposition \ref{alphas} $\psi $ is a core-allocation.
\end{proof}

\begin{corollary}
Let $(N,D,Z)$ be a SI-situation and $(N,c)$ the corresponding SI-game. If
the reduced SI-situation $(N,D,\widetilde{Z})$ is consolidated, then $\psi
\in Core(N,c)$.
\end{corollary}

From now on, the allocation $\psi$ will be called the extended Owen point
for consolidated SI-games (those which come from consolidated
SI-situations). \medskip

Recall that according to Sprumont (1990), a population monotonic allocation
scheme (pmas), for the game $(N,c)$ is a collection of vectors $y^{S}\in 
\mathbb{R}^{s}$ for all $S\subseteq N,S\neq \varnothing $ {} such that $%
y^{S}(S)=c(S)$ for all $S\subseteq N,S\neq \varnothing ,$ and $y_{i}^{S}\geq
y_{i}^{T}$ for all $S\subseteq T\subseteq N$ and $i\in S.$ The reader may
note that whenever $\left( y^{S}\right) _{\varnothing \neq S\subseteq N}$ is
a pmas for $(N,c),$ then $y^{S}$ is a core allocation for the game $(S,c_s)$
for all $S\subseteq N,S\neq \varnothing .$ Thus, cost allocations attainded
through a pmas are a refinement of the core. This implies that every cost TU
game with a pmas is totally balanced but the reciprocal is not true and
there are many totally balanced cost TU games without pmas. (A
core-allocation for $(N,c),$ i.e. $x\in Core(N,c)$, is reached through a
pmas if there exists $\left( y^{S}\right) _{\varnothing \neq S\subseteq N}$
for the game $(N,c)$ such that $y_{i}^{N}=x_{i}$ for all $i\in N.$) \medskip

The final result of the section explicitly constructs a pmas that realizes
the extended Owen point for consolidated SI-games.

\begin{theorem}
Let $(N,D,Z)$ be a consolidated SI-situation and $(N,c)$ the corresponding
SI-game. Then, $\psi $ can be realized through a pmas.
\end{theorem}

\begin{proof}
Define for all $i\in S,S\subseteq N$ and $S\neq \varnothing ,$

\begin{equation*}
y_{i}^{S}:=\sum_{t=1}^{T}P_{t}^{S}(\sigma ^{S})d_{t}^{i}+\sum_{t\in \Upsilon
^{S}/i\in J_{t}^{S}}\frac{k_{t}^{S}}{\left\vert J_{t}^{S}\right\vert }.
\end{equation*}

\noindent where $\Upsilon ^{S}:=\{t\in T|\exists i\in S\mbox{ such
that }\delta _{t}(\sigma ^{S})=1$ $\mbox{ for
all }S\subseteq N\text{ with }i\in S\}$ and $J_{t}^{S}:=\{i\in S$ such that $%
\delta _{t}(\sigma ^{i})=1\}.$ Then for all $S\subseteq N,S\neq \varnothing $

\begin{equation*}
\sum_{i\in S}y_{i}^{S}=\sum_{t=1}^{T}P_{t}^{S}(\sigma
^{S})d_{t}^{S}+\sum_{t\in T^{S}}k_{t}^{S}=c(S),
\end{equation*}%
and for all $S\subseteq R\subseteq N,S,R\neq \varnothing $ and for all $i\in
S,$

\begin{eqnarray*}
y_{i}^{S} &=&\sum_{t=1}^{T}P_{t}^{S}(\sigma ^{S})d_{t}^{i}+\sum_{t\in
\Upsilon ^{S}/i\in J_{t}}\frac{k_{t}^{S}}{\left\vert J_{t}^{S}\right\vert }%
\geq \sum_{t=1}^{T}P_{t}^{R}(\sigma ^{R})d_{t}^{i}+\sum_{t\in \Upsilon
^{S}/i\in J_{t}^{S}}\frac{k_{t}^{R}}{\left\vert J_{t}^{S}\right\vert } \\
&\geq &\sum_{t=1}^{T}P_{t}^{R}(\sigma ^{R})d_{t}^{i}+\sum_{t\in \Upsilon
^{R}/i\in J_{t}^{R}}\frac{k_{t}^{R}}{\left\vert J_{t}^{R}\right\vert }%
=y_{i}^{R},
\end{eqnarray*}%
since $\Upsilon ^{R}\subseteq \Upsilon ^{S}$ and $\left\vert
J_{t}^{S}\right\vert \leq \left\vert J_{t}^{R}\right\vert $ for all $t\in T$.

Finally, we see that $y_{i}^{N}=\psi _{i}$ for all $i\in N.$ So, the
extended Owen point for consolidated SI-situations $\psi $ can be reached
through the pmas $\left( y^{S}\right) _{\varnothing \neq S\subseteq N}.$
\end{proof}

\medskip

From the proof of the above theorem we deduce that for every consolidated
SI-game, a pmas can be built just taking the extended Owen point for each
subgame and gathering them all as a collection of vectors. Notice that this
construction shows a strong consistency, in terms of stability, of this
point solution. \medskip

The final example illustrates all the above mentioned results. In addition,
it shows that the core of consolidated SI-games is not necessarily a
singleton.

\begin{example}
Consider the following SI-situation with three periods and three players:%
\begin{equation*}
\begin{array}{|c|c|c|c||c|c|c||c|c|c||c|c|c||c|c|c|}
\hline
& \multicolumn{3}{|c||}{Demand} & \multicolumn{3}{|c||}{Purchasing} & 
\multicolumn{3}{|c||}{Inventory} & \multicolumn{3}{|c||}{Backlogging} & 
\multicolumn{3}{|c|}{Setup} \\ \hline
P1 & $1$ & $3$ & $1$ & $1$ & $1$ & $1$ & $1$ & $1$ & $1$ & $1$ & $1$ & $1$ & 
$1$ & $1$ & $5$ \\ \hline
P2 & $2$ & $1$ & $1$ & $2$ & $3$ & $4$ & $1$ & $1$ & $1$ & $1$ & $1$ & $1$ & 
$1$ & $1$ & $5$ \\ \hline
P3 & $2$ & $1$ & $3$ & $2$ & $3$ & $5$ & $1$ & $1$ & $1$ & $1$ & $1$ & $1$ & 
$1$ & $1$ & $5$ \\ \hline
\end{array}%
\end{equation*}
The corresponding SI-game is given in the next table:%
\begin{equation*}
\begin{array}{|c|c|c|c||c|c|c||c|c|c||c|c|c||c|c|c||c||}
\hline
& d_{1}^{_{S}} & d_{2}^{_{S}} & d_{3}^{_{S}} & p_{1}^{_{S}} & p_{2}^{_{S}} & 
p_{3}^{_{S}} & h_{1}^{_{S}} & h_{2}^{_{S}} & h_{3}^{_{S}} & b_{1}^{_{S}} & 
b_{2}^{_{S}} & b_{3}^{_{S}} & k_{1}^{_{S}} & k_{2}^{_{S}} & k_{3}^{_{S}} & c
\\ \hline
\{1\} & 1 & 3 & 1 & 1 & 1 & 1 & 1 & 1 & 1 & 1 & 1 & 1 & 1 & 1 & 5 & 8 \\ 
\hline
\{2\} & 2 & 1 & 1 & 2 & 3 & 4 & 1 & 1 & 1 & 1 & 1 & 1 & 1 & 1 & 5 & 12 \\ 
\hline
\{3\} & 2 & 1 & 3 & 2 & 3 & 5 & 1 & 1 & 1 & 1 & 1 & 1 & 1 & 1 & 5 & 20 \\ 
\hline
\{1,2\} & 3 & 4 & 2 & 1 & 1 & 1 & 1 & 1 & 1 & 1 & 1 & 1 & 1 & 1 & 5 & 13 \\ 
\hline
\{1,3\} & 3 & 4 & 4 & 1 & 1 & 1 & 1 & 1 & 1 & 1 & 1 & 1 & 1 & 1 & 5 & 17 \\ 
\hline
\{2,3\} & 4 & 2 & 4 & 2 & 3 & 4 & 1 & 1 & 1 & 1 & 1 & 1 & 1 & 1 & 5 & 31 \\ 
\hline
\{1,2,3\} & 5 & 4 & 5 & 1 & 1 & 1 & 1 & 1 & 1 & 1 & 1 & 1 & 1 & 1 & 5 & 22
\\ \hline
\end{array}%
\end{equation*}

The reader may notice that it comes from a consolidate SI-situation since%
\begin{equation*}
\begin{array}{|c|c|c|c||c|c|c||c||}
\hline
& \sigma _{1}^{S} & \sigma _{2}^{S} & \sigma _{3}^{S} & P_{1}^{S}(\sigma
^{S}) & P_{2}^{S}(\sigma ^{S}) & P_{3}^{S}(\sigma ^{S}) & \delta (\sigma
^{S})^{\prime }k^{S} \\ \hline
\{1\} & 2 & 2 & 2 & 2 & 1 & 2 & 1 \\ \hline
\{2\} & 1 & 1 & 1 & 2 & 3 & 4 & 1 \\ \hline
\{3\} & 1 & 1 & 1 & 2 & 3 & 4 & 1 \\ \hline
\{1,2\} & 1 & 2 & 2 & 1 & 1 & 2 & 2 \\ \hline
\{1,3\} & 1 & 2 & 2 & 1 & 1 & 2 & 2 \\ \hline
\{2,3\} & 1 & 1 & 1 & 2 & 3 & 4 & 1 \\ \hline
\{1,2,3\} & 1 & 2 & 2 & 1 & 1 & 2 & 2 \\ \hline
\end{array}%
\end{equation*}%
The extended Owen point for the above consolidated SI-game is $\psi =\left(
7,\frac{11}{2},\frac{19}{2}\right) .$ However, the core of this game does
not reduce to it since also $x=(7,5,10)\in Core(N,c).$ In addition, the
extended Owen point can be reached through the pmas 
\begin{equation*}
\left( \left( 8\right) ^{\{1\}},\left( 12\right) ^{\{2\}},\left( 20\right)
^{\{3\}},\left( 7,6\right) ^{\{1,2\}},\left( 7,10\right) ^{\{1,3\}},\left( 
\frac{23}{2},\frac{39}{2}\right) ^{\{2,3\}},\left( 7,\frac{11}{2},\frac{19}{2%
}\right) ^{\{1,2,3\}}\right) .
\end{equation*}
\end{example}

\section{Concluding Remarks}

Cooperation in periodic review finite horizon inventory models has been
already analyzed in Guardiola et al. 2008, 2009 and Van Den Heuvel et al.
2007. This paper extends previous approaches in the literature considering a
more general model that includes non-homogeneous set up and backlogging
costs. We prove that this model of cooperation, by sharing technologies for
the production, carrying of goods and distribution channels, induces savings
because the resulting game is totally balanced. Moreover, we have introduced
a parametric family of allocations based on the Owen point (see Guardiola et
al. 2008, 2009) and a subclass of games that enjoys a population monotonic
allocation scheme.

The stability property of the above mentioned mode of coordination leads us
to mention two related future research lines: (1) analyzing the cooperation
aspects of broader subclasses of inventory situations for which it is
possible to provide explicit solutions; and (2) studying the relationships
between the cores that arise from situations with and without set up costs.

\section{Acknowledgments}

The research of the second author is partially supported Financial support
of the Ministerio de Ciencia, Innovaci\'{o}n y Universidades (MCIU), the
Agencia Estatal de Investigaci\'{o}n (AEI) and the Fondo Europeo de
Desarrollo Regional (FEDER) under the project {PGC2018-097965-B-I00}. The
research of the third author has been partially supported by Spanish
Ministry of Education and Science/FEDER grant number {MTM2016-74983-C02-01},
and projects {FEDER-US-1256951}, {CEI-3-FQM331} and \textit{NetmeetData}:
Ayudas Fundaci\'{o}n BBVA a equipos de investigaci\'{o}n cient\'{\i}fica
2019.


\begin{thebibliography}{99}
\bibitem{B63} Bondareva ON (1963) Some applications of linear programming
methods to the theory of cooperative games.\ Problemy Kibernety 10\textbf{:}%
119-139. In Russian.

\bibitem{BHH01} Borm PEM, Hamers H and Hendrickx R (2001) Operations
Research Games: A Survey.\ TOP 9\textbf{:}139-216.

\bibitem{CZ07} Chen X and Zhang J (2007) Duality approaches to economic lot
sizing games.\ Production and Operations Management 25(7):1203-1215.

\bibitem{GPRE00} {\ Gellekom JRG, Potters JAM, Reijnierse JH, Engel MC and
Tijs SH (2000) Characterization of the Owen Set of Linear Production
Processes.\ Games and Economic Behavior }32{:139-156.}

\bibitem{Dr10} Dreschel J (2010) Cooperative Lot Sizing Games in Supply
Chains{.\ }Springer-Verlag Berlin Heidelberg.

\bibitem{GOA11} Gopaladesikan M and Uhan NA (2011){\ }Cost Sharing for the
Economic Lot-Sizing Problem with Remanufacturing Options.{\ }%
Optimization-on-line.org/DB\_FILE/2010/09/2733.

\bibitem{GMJ04a} {Guardiola LA, Meca A, Puerto J (2009) Production-Inventory
games: a new class of totally balanced combinatorial optimization games.\
Games and Economic Behavior }65{:205-219}.

\bibitem{GMJ04b} {Guardiola LA, Meca A, Puerto J (2008) PI-games and pmas
games: characterizations of Owen point. }Mathematical Social Sciences
56:96-108.

\bibitem{JM74} Johnson LA and Montgomery DC (1974) Operations Research in
Production Planning, Scheduling, and Inventory Control. John Wiley \& Sons.

\bibitem{Ka03} Karimi B, Ghomi SMTF and Wilson JM (2003) The capacitated lot
sizing problem: a review of models and algorithms.\ Omega 31\textbf{:}%
365-378.

\bibitem{LI14} Li GD, Du DL, Xu DC and Zhang RY (2014) A cost-sharing method
for the multi-level economic lot-sizing game{.\ }Science China Information
Sciences volume 57:{1-9}.

\bibitem{O75} {Owen G (1975) On the core of linear production games.\
Mathematical Programming }9{:{3}58-370.}

\bibitem{Perea2009} Perea F, Puerto J, and Fern\'{a}ndez FR (2009) Modeling
cooperation on a class of distribution problems.\ European Journal of
Operational Research 198\textbf{(}3\textbf{):}726-733.

\bibitem{Perea2012} Perea F, Puerto J, and Fern\'{a}ndez FR (2009) Avoiding
unfairness of Owen allocations in linear production processes. European
Journal of Operational Research 220:125-131.

\bibitem{SH67} Shapley LS (1967) On Balanced Sets and Cores.\ Naval Res.
Logist.\textbf{\ }14\textbf{:}453-460.

\bibitem{SS69} Shapley LS and Shubik M (1969) On market games.\ Journal of
Economic Theory 1:9-25.

\bibitem{S90} Sprumont Y (1990) Population Monotonic Allocation Schemes for
Cooperative Games with Transferable Utility.\ Games and Economic Behavior
2:378-394.

\bibitem{TU14} Toriello A and Uhan NA (2014) Dynamic Cost Allocation for
Economic Lot Sizing Games.\ Operations Research Letters 42(1):82-84.

\bibitem{TSA13} Tsao YC, Chen TH and Wu PY (2013) Effects of Lot-Sizing
Integration and Learning Effect on Managing Imperfect Items in a
Manufacturer-Retailer Chain.\ Journal of Applied Mathematics 9:1-11.

\bibitem{HBH05} Van den Heuvel W, Borm P and Hamers H (2007) Economic
lot-sizing games.\ European Journal of Operational Research 176:1117-1130.

\bibitem{XuYA09} Xu D and Yang R (2009) A cost-sharing method for an
economic lot-sizing game.\ Operations Research Letters 37:107-110.

\bibitem{ZA69} Zangwill WI (1969) A backlogging model and multi-echelon
model of a dynamic economic lot size production system- a network approach.\
Management Science 15(9):506-527.

\bibitem{ZENG11} Zeng Y, Li J and Cai X (2011) Economic lot-sizing games
with perishable inventory.\ ICSSSM11, Tianjin,  pp. 1-5, doi:
10.1109/ICSSSM.2011.5959533.
\end{thebibliography}
\end{document}